\documentclass[11pt]{amsart}

\usepackage{epsfig,amsfonts,amssymb,amsmath,amsthm,eucal}
\usepackage{graphicx}
\usepackage{hyperref}

\newtheorem{thm}{Theorem}[section]
\newtheorem{Theorem}{Theorem}[section]

\newtheorem{Lemma}[thm]{Lemma}
\newtheorem{prop}[thm]{Proposition}
\newtheorem{Proposition}[thm]{Proposition}
\newtheorem{cor}[thm]{Corollary}

\theoremstyle{remark}
\theoremstyle{remark}

\newcommand{\Z}{\mathbb Z}

\newcommand{\R}{\mathbb R}

\newcommand{\e}{\varepsilon}

\newcommand{\oh}{\tfrac12}

\newcommand{\bz}{{\mathbb Z}}
\newcommand{\br}{{\mathbb R}}

\begin{document}

\title[Piecewise periodic]{Random skew plane partitions with a piecewise periodic back wall}

\author[C. Boutillier]{Cedric Boutillier}
\address{Cedric Boutillier\newline UPMC Univ Paris 06\newline UMR 7599, LPMA\newline F-75005, Paris, France 
\newline\&\newline
D\'{e}partement de Math\'{e}matiques et Applications (DMA),
\newline \'{E}cole Normale Sup\'{e}rieure
\newline 75230 Paris, FRANCE
}
\email{cedric.boutillier@upmc.fr}

\author[S. Mkrtchyan]{Sevak Mkrtchyan}
\address{Sevak Mkrtchyan \newline Rice University, Math Department - MS136\newline 6100 S. Main St,\newline Houston, TX, 77005}
\email{sevak.mkrtchyan@rice.edu}

\author[N. Reshetikhin]{Nicolai Reshetikhin}
\address{Nicolai Reshetikhin \newline UC Berkeley, Department of Mathematics\newline 970 Evans Hall 3840,\newline Berkeley, CA 94720}
\email{reshetik@math.berkeley.edu}

\author[P. Tingley]{Peter Tingley}
\address{Peter Tingley \newline MIT dept. of math\newline 77 Massachusetts Ave\newline Cambridge, MA, 02139}
\email{ptingley@math.mit.edu}

\thanks{The third and fourth authors were partially supported by the NSF grant DMS-0354321.}

\begin{abstract}

Random skew plane partitions of large size distributed according to an appropriately scaled 
Schur process develop limit shapes. In the present work we consider the limit of large random skew plane partitions where the inner boundary approaches a piecewise linear curve 
with non-lattice slopes, describing the limit shape and the local fluctuations in various regions. 
This analysis is fairly similar to that in \cite{OR2}, but we do find some new behavior. For instance,
the boundary of the limit shape 
is now a single
smooth (not algebraic) curve, whereas the boundary in \cite{OR2} is singular. 
We also observe the bead process introduced in \cite{beads} appearing in the 
asymptotics at the top of the
limit shape.
\end{abstract}

\maketitle

\tableofcontents

\section{Introduction}
\label{sec:intro}

A \emph{skew plane partition} with inner boundary being given by a partition $\lambda$ is any array of positive integers $\pi= \{\pi_{i,j} \}$ defined for pairs $(i,j)$ such that $i\geq 1$ and $j\geq \lambda_i$,
which is non-increasing in both $i$ and $j$.
A skew plane partition is called {\it bounded} by $N$ and $M$ if $\pi_{ij}=0$ when $i>N$ or $j>M$.
An example of a skew plane partition with inner shape $\lambda=(3,2,2)$ is given on Figure \ref{fig:skew_partition}. Looking at this picture we can also see that a bounded skew plane partition can be regarded as monotonic piles of cubes in a certain ``semi-infinite room'' with no roof (see also Figure \ref{fig:empty}).

A natural question to ask is: what is the shape of a typical pile containing a
large number of cubes? One must first make this precise by fixing a probability distribution. 
One natural distribution to choose would be the uniform distribution on all 
skew plane partitions with a fixed number
of cubes.
This is actually quite difficult to deal with, and we instead
consider the following probability measure (\emph{grand-canonical ensemble}) on bounded skew plane partitions:
\begin{equation}\label{eq:distr}
P(\pi)=\frac{1}{Z} q^{|\pi|},
\end{equation}
where $|\pi|=\sum_{i,j}\pi_{i,j}$ is the total number of boxes in the piles corresponding to the plane partition $\pi$, $0<q<1$ and
\[
Z=\sum_{\pi} q^{|\pi|}
\]
is the normalization factor known as the \emph{partition function} in
statistical mechanics.
The larger $q$ is, the more likely it is to have many
cubes. We will study certain limits of this system as $q$ approaches $1$, the bounds $M$ and $N$ approach infinity, and the partition $\lambda$ defining the back
wall of the room grows in a specified way.

There are many ways to take such a limit. For instance, one can specify that, after an appropriate rescaling, 
$\lambda$
approaches a fixed curve. The limiting behavior of the system will certainly depend on this curve. Some such cases have previously been studied:
\begin{itemize}
  \item The case when $\lambda$ is empty was studied in the unbounded case in \cite{CerfKenyon} and in the bounded case in
    \cite{OR1}.
  \item The case when $\lambda$ approaches a piecewise linear curve with
    lattice slopes was studied in \cite{OR2}.
\end{itemize}
In this article we consider a limit where $\lambda$ approaches a piecewise
linear function with a (non-lattice) slope staying \emph{strictly} between $-1$ and $1$. The
large size limit is studied via \emph{correlation functions} using the
formalism and notations developed in \cite{OR1,OR2}. 

Before we start, let us mention some other related work. Using results of this article, the limit where $\lambda$ approaches a piecewise linear function with arbitrary slopes has been studied in \cite{M}. Local fluctuations and the free energy for the partition function on a torus were first computed in \cite{NHB}.

\subsection{Random skew plane partitions and Correlation functions}

Our piles of unit cubes are in $\R^3$ with coordinates $(x,y,z)$
centered at points $(\Z+\frac{1}{2},\Z+\frac{1}{2},\Z+\frac{1}{2})$ and bounded
in $(x,y)$ directions as it was described above. That is, the $(x,y)$ coordinates of centers of cubes satisfy conditions:
\begin{equation*}
  \begin{cases}
    \frac{1}{2} \leq x \leq N-\frac{1}{2},\\
    \lambda_{x+1/2} \leq y \leq M.
  \end{cases}
\end{equation*}
The partition $\lambda $ is given and describes the configuration of the back wall. Denote this region of $\R^3$ by $D_{\lambda, M,N}$.

\begin{figure}
\includegraphics[width=8cm]{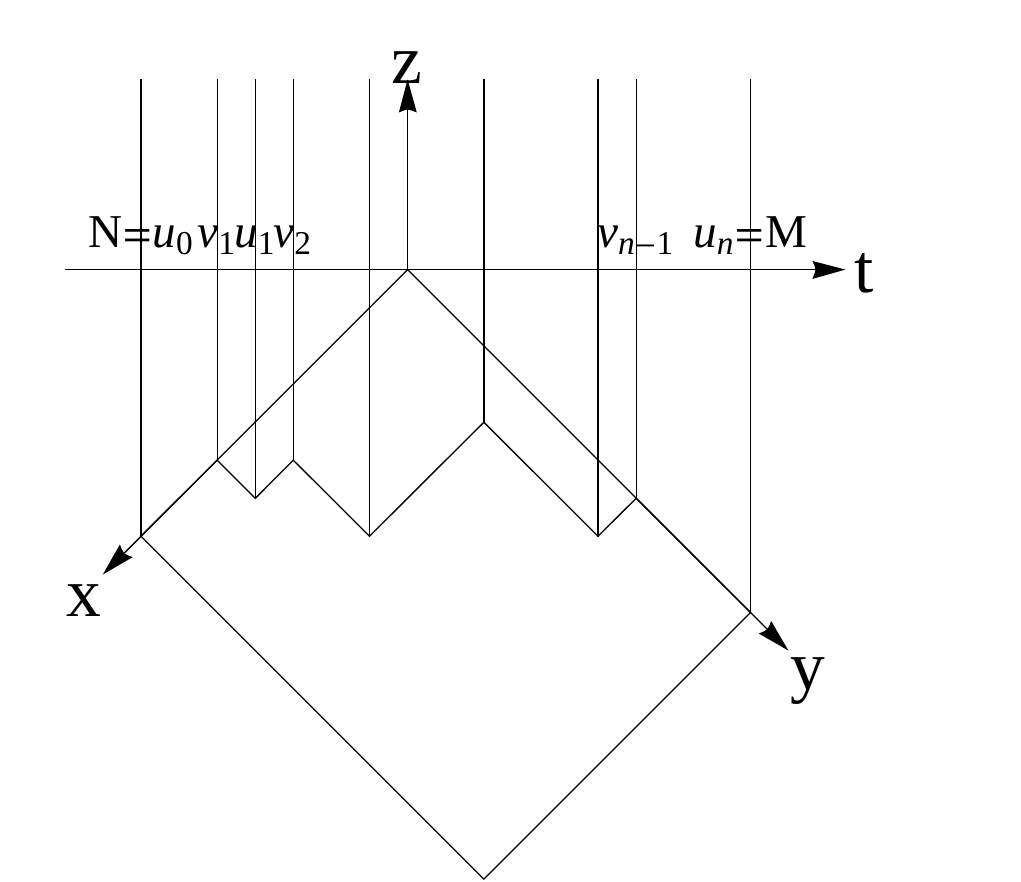}
\caption{ \label{fig:empty} An empty skew plane partition. Here the inner shape is given by the partition $(6,3,3,3,1)$.}
\end{figure}

The mapping 
\begin{equation*}
  (x,y,z)\mapsto \bigl(t=y-x,h=z-\frac{x+y}{2}\bigr)
\end{equation*}
projects $\R^3$ to $\R^2$.
This projection maps piles of cubes to tilings of the region
\begin{equation*}
E_{M, N}=\bigl\{(t,h)\ |\ -M\leq t\leq N;
h\geq \max\{-M-t/2, N+t/2\}\bigr\}
\end{equation*}
which differ from the one corresponding to the empty pile in finitely many places. This is evident from Fig. \ref{fig:skew_partition}.

\begin{figure}
\includegraphics[width=5cm]{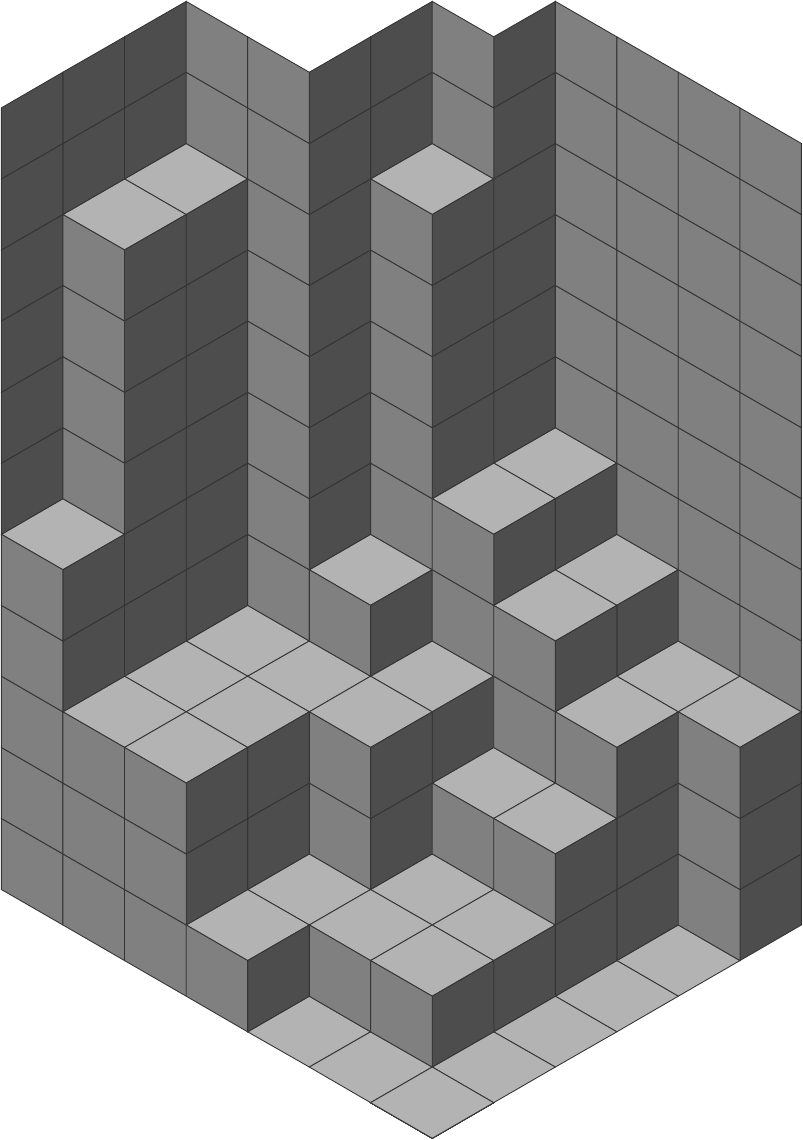}
\caption{  A skew plane partition represented as a stepped surface. Notice that the same picture can be viewed as a tiling of a portion of the plane by three types of rhombi.
}
\label{fig:skew_partition}
\end{figure}

A convenient characterization of a system of random skew plane partitions is
the collection of local correlation functions $\rho_{(t_1,h_1),\dots(t_k,h_k)}$ giving the probability to see horizontal rhombi at $(t_1,h_1),\dots,(t_k,h_k)$:
\[
\rho_{(t_1,h_1),\dots, (t_k, h_k)}=\sum_{\pi} P(\pi)\prod_{a=1}^k\sigma_{i_a,j_a}(\pi)
\]
where $\sigma_{t,h}(\pi)=1$ when $(t,h)$ is the position
of the center of a horizontal rhombus and $\sigma_{t,h}(\pi)=0$
if not.

An explicit formula was found in \cite{OR1,OR2} for such correlation functions.
To describe this formula, we first encode the boundary of the partition $\lambda$ as a function $h=-b(t)$, where
\begin{equation} \label{eq:bt}
b(t)=\frac{1}{2}\sum_{i=1}^n|t-v_i|-\frac{1}{2}\sum_{i=1}^{n-1}|t-u_i|.
\end{equation}
The points $u_1,\dots,u_{n-1}$ (resp. $v_1,\dots,v_n$) represent the inner (resp. the outer) corners of the back wall. See Figure \ref{fig:empty}.
It is convenient to set $u_0=-M$ and $u_n=N$. The lines $t=u_0$ and $t=u_n$ are limiting on the left and on the right the region tiled by rhombi.

Then we introduce the partition of the horizontal interval  $\Z + \frac{1}{2} \cap (u_0,u_n)$ into $D_-\cup D_+$ where
\begin{align*}
  D^+ &=\bigl\{m  \in \bz + \frac{1}{2} \cap (u_0,u_n)\ |\  \text{where the boundary $\lambda$ in $(t,h)$ has slope $-\frac{1}{2}$ at  $m$} \}, \\
  D^- &=\bigl\{m  \in \bz + 1/2 \cap (u_0,u_n)\ |\ \text{where the boundary $\lambda$ in $(t,h)$ has slope $+\frac{1}{2}$ at } m \bigr\}.
\end{align*}

Define the functions
$\Phi_\pm(z,t)$ by:
\begin{equation*}
  \Phi_+(z,t)=\prod_{\substack{m\in D^+\\ m > t }}(1-zq^m), \quad
  \Phi_-(z,t)=\prod_{\substack{ m\in D^-\\ m < t  }}(1-z^{-1}q^{-m}).
\end{equation*}

The following holds:
\begin{thm} \cite{OR2} \label{thm:OR2_main}
Correlation functions are determinants of the correlation kernel:
\begin{equation}\label{eq:det}
\rho_{(t_1,h_1),\dots, (t_k, h_k)}=\det \bigl[K\bigl((t_i,h_i),(t_j,h_j)\bigr)\bigr]_{1\leq i, j \leq k}
\end{equation}
The correlation kernel $K$ is given by
\begin{multline}\label{eq:main-corr2}
K((t_1,h_1),(t_2,h_2))= \frac{1}{(2\pi i)^2}
\oint_{z\in C_1}\oint_{w\in C_2}
\frac{\Phi_-(z,t_1)\Phi_+(w,t_2)}{\Phi_+(z,t_1)\Phi_-(w,t_2)}\\
\times\frac{\sqrt{zw}}{z-w}z^{-h_1-b(t_1)-1/2}w^{h_2+b(t_2)-1/2}\frac{dzdw}{zw},
\end{multline}
where $C_1$  is a simple positively oriented contour around $0$ such that its interior contains none of the zeroes of $\Phi_+(z,t_1)$ . Similarly $C_2$ is a simple positively oriented contour around $0$ such that its exterior does not contain
zeroes of $\Phi_-(w,t_2)$. Moreover, if $t_1< t_2$, then $C_1$ is contained in the interior of $C_2$, and otherwise, $C_2$ is contained in the interior of $C_1$.
\end{thm}

\subsection{Main results} 
\label{subsec:main_intro}

We are interested in the behavior of random skew plane partitions sampled according to the probability distribution \eqref{eq:distr} for \emph{large domains}, meaning that a limit of the following form is taken:
\begin{itemize}
\item Fix a sequence $(r_k)$ of positive real numbers such that $\lim_{k\to \infty} r_k=0$ and for all $k$, let $q_k=\exp(-r_k)$.
\item Choose  sequences $(M_k)$, $(N_k)$ and $(\lambda_k)$ such that as $k$ goes to infinity, the 3-dimensional domain $r_k D_{\lambda_k, M_k, N_k}$ converges to some region $D$ of $\R^3$.
\item Find characteristic scales of fluctuations for large $k$ and describe random plane partitions at appropriate scales in the limit $k\to \infty$.
\end{itemize}

The increasing sequence of domains where
$r_k =\frac{1}{k}$ and the corners $u^{(k)}_i$ and $v^{(k)}_i$ describing the partition $\lambda_k$ are given by:
\begin{equation*}
  u^{(k)}_i=kU_i, \quad v^{(k)}_i=kV_i,
\end{equation*}
where $n$ and $U_0 \leq V_1 \leq \dots \leq V_n \leq U_n$ are fixed integers was studied in \cite{OR2}.
As $k\to \infty$, the function $(x,y)\mapsto h_k(x,y)=\pi_{kx,ky}/k$ converges in probability to a fixed limit shape.
This function is infinitely differentiable everywhere except along a curve
where the surface it defines merges with the walls. Along the curve, its second
derivative is discontinuous (the first derivative has a square root
singularity).  The projection of this curve to the $(t,h)$-plane is called
the \emph{arctic circle}  or
the boundary of the frozen region. Indeed outside this curve the tiling by
rhombi is ``frozen", which is to say that with probability tending to 1, the
tiles outside of the curve are of only one type.

The results in \cite{OR2} are obtained by using the steepest descent method to evaluate the integrals in Equation (\ref{eq:main-corr2}). In the limit one obtains exact formulas for the limit shape, as well as descriptions of the behavior of local correlation functions.

In this paper, we focus on a different type of increasing sequences of domains. We keep $r_k$ equal to $\frac{1}{k}$, but we assume that the partitions $\lambda_k$ describing the back wall consist of $n$ segments growing in length, and that each of the segments is not straight anymore, but is a ``staircase'' with fixed slope. To be precise, the $j^{th}$ segment of $\lambda_k$ consists of $k L_j$ copies of the corner with sides $(a_j, b_j)$ for some $L_j,a_j, b_j \in \Z_{>0}$ fixed, see Figure \ref{fig:NotationGraph}. 

In the framework described above, we compute the limit shape of the random pile of boxes as $k\to \infty$ and correlation functions. We show that:
\begin{itemize}
\item The system rescaled by a factor $r_k$ converges to a deterministic limit shape. This is a smooth transcendental function in the bulk. The frozen boundary consists of a single smooth curve.
\item As in \cite{OR2}, correlation functions in the bulk are given by the incomplete Beta kernel.
\item For most values of $t$, if one fixes $t$ and allows $h$ to approach infinity, the slope of the limit shape turns vertical, and the determinantal process with the incomplete Beta kernel describing the statistics of horizontal tiles degenerates into the bead process introduced in \cite{beads}.
\end{itemize}

\subsection{Relation to Dimer models} The limit shapes for rhombi tilings of bounded domains with piecewise linear boundary parallel to axes $h=const, h\pm t/2=const$ were studied in \cite{KO} where it was shown that uniformly distributed tilings develop a limit shape given by an algebraic curve. This generalizes the arctic circle theorem for large hexagons \cite{CohnLarsenPropp,J}. Local correlation functions for large hexagonal domains with the uniform distribution of rhombi tilings were studied in \cite{J,G}, and in \cite{Ke:GFF2}. 

The rhombi used in these tilings can be thought of as \emph{dimers}, since they are the union of two adjacent faces of the triangular lattice (\emph{monomers}). These random tilings are thus examples of \emph{dimer models}.
Dimer models have a long history in statistical mechanics, see for example \cite{Ka,TempFish,MCW}. Corresponding tiling models with height functions were first studied in the physics literature, where many ideas about the limit shape and the
structure of correlation functions in the thermodynamical limit can be traced, see for example \cite{NHB}.

\subsection{Acknowledgments} We are grateful to A. Borodin, A. Okounkov, and P. van Moerbeke for interesting discussions.

\section{Setup: a piecewise periodic back wall}
We are studying a system of random skew plane partitions, as described in Section \ref{sec:intro}. The precise system is described in Section \ref{subsec:main_intro}. Let us start by introducing some more notation. 

Let $\tau = rt$ and $\chi = rh$ be the rescaled coordinates. The $\tau$ coordinates where the slope of the back wall changes are denoted by $A_0, A_1,\dots ,A_n$, where $A_0$ and $A_n$ are the bounds of the shape. The $t$ coordinates that correspond to these corners will be $A_1/r,\ldots,A_n/r$. Between $A_{k-1}/r$ and $A_k/r$ we have a succession of patterns made of $a_k$ ascending steps and $b_k$ descending steps, repeated $L_k/r$ times. The slope of the piece of wall between $A_{k-1}$ and $A_k$
is $\alpha_k = \frac{a_k-b_k}{2(a_k+b_k)}$.
It will be convenient to define $\alpha_0=-\frac{1}{2}$ and $\alpha_{n+1}=\frac{1}{2}$.
\begin{figure}
  \begin{center}
  \includegraphics[width=14.6cm]{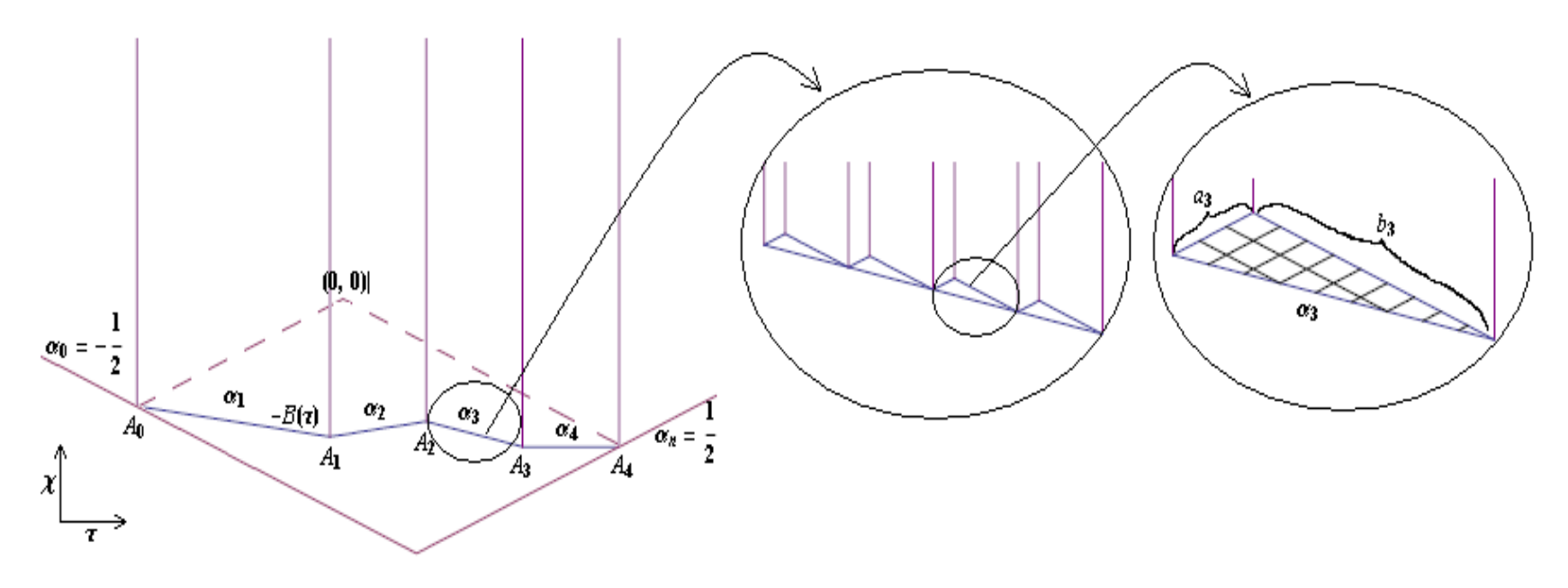}
  \end{center}
  \caption{The setup   }
\label{fig:NotationGraph}
\end{figure}
Let $j$ be the integer in $\{1,\dots,n\}$ such that $\tau \in (A_{j-1},A_j]$.
Define the function $B(\tau)$ by
\begin{equation}\label{eq:B}
 -B(\tau) = -\frac{1}{2}A_0+ \sum_{k=1}^{j-1}\alpha_k (A_{k}-A_{k-1}) + \alpha_j  (\tau-A_{j-1}) = -\sum_{k=0}^{j-1} A_k (\alpha_{k+1} -\alpha_k) + \alpha_j \tau,
\end{equation}
if $ A_{j-1} <\tau < A_j$.
$B(\tau)$ is the limiting rescaled boundary function $b(t)$ of Equation~\eqref{eq:bt} expressed in the rescaled variables $\tau$ and $\chi$. We can decide to define $B(\tau)$ from the rightmost corner $A_n$ instead of the leftmost one $A_0$, yielding an alternative expression:
\begin{equation}\label{eq:B2}
  -B(\tau) = \frac{1}{2}A_n -\sum_{k=j+1}^{n}\alpha_k(A_{k}-A_{k-1}) -\alpha_j(A_{j}-\tau)
  = \alpha_j \tau + \sum_{k=j}^{n}(\alpha_{k+1}-\alpha_k) A_k.
\end{equation}

We are interested in the correlation functions in the scaling limit when
$\lim_{r\rightarrow +0} r t_1=\lim_{r\rightarrow +0}r t_2 \text{ and } t_1-t_2=\Delta(t) \text{ is a constant}.$

\subsection{The function \texorpdfstring{$S_{\tau,\chi}(z)$}{S(z)}} We want to use the steepest descent method in order to find the asymptotics of local correlation functions as the size of the system increases and $r\to 0$. The first step in doing this is to rewrite the integral defining the correlation kernel in Theorem \ref{thm:OR2_main} as
\begin{equation*}
  K\bigl( (t_1,h_1),(t_2,h_2) \bigr) = \frac{1}{(2i\pi)^2} \oint \oint e^{\frac{S^{(r)}_{t_1,h_1}(z)-S^{(r)}_{t_2,h_2}(w)}{r}} \frac{1}{z-w} \frac{dz}{2i\pi z} \frac{dw}{2i\pi w},
\end{equation*} show that $S^{(r)}_{t,h}(z)$ converges to some function $S_{\tau,\chi}(z)$ as $r\to 0$, $r t\to \tau$, $r h\to \chi$ and find the critical points of $z\mapsto S_{\tau,\chi}(z)$.

The function $S^{(r)}_{t,h}(z)$ is given by:
 \begin{multline}
   S^{(r)}_{t,h}(z)= r\log \left(z^{-h-b(t)}\frac{\Phi_-(z,t)}{\Phi_+(z,t)}\right) \\
 =-r(h+b(t))\log z + \sum_{\stackrel{m<t, m\in D^-,}{ m\in \Z+\oh}}r\log(1-z^{-1}e^{rm}) -\sum_{\stackrel{m>t, m\in D^+,}{ m\in \Z+\oh}}r\log(1-ze^{-rm}).
\end{multline}
Here and in the sequel, $\log$ denotes the branch of the logarithm with an imaginary part in $(-\pi,\pi)$, with a cut along $\mathbb{R}_{-}$.

In order to find the limit $S_{\tau,\chi}(z)$ of this quantity, we first state a lemma:

\begin{Lemma} \label{lem:CD}
If $[C,D]$ corresponds to a piece of wall with parameters $a,b$ then
\begin{equation} \label{eq:oon}  \lim_{r\rightarrow 0}
	\sum_{\stackrel{\tfrac Cr<m'<\tfrac Dr, m'\in D^-,}
			{m'\in\Z+\oh}}
		r\log(1-z^{-1}e^{rm'})=\frac{a}{a+b}
	\int_{C}^{D} \log (1-z^{-1}e^v)dv \end{equation}
and
\begin{align}
\lim_{r\rightarrow 0}
	\sum_{\stackrel{\tfrac Cr<m'<\tfrac Dr, m'\in D^+,}
			{m'\in\Z+\oh}}
		r\log(1-z e^{-rm'})=
\frac{b}{a+b}
	\int_{C}^{D} \log (1-z e^{-v})dv.
\end{align}
\end{Lemma}

\begin{proof}
The proof is elementary. Here is how it goes for \eqref{eq:oon}:
\begin{align*}
\quad&
\lim_{r\rightarrow 0}
	\sum_{\stackrel{\tfrac Cr<m'<\tfrac Dr, m'\in D^-,}
			{m'\in\Z+\oh}}
		r\log(1-z^{-1}e^{rm'})=
\lim_{r\rightarrow 0}
	\sum_{\stackrel{C<m<D, m\in D^-,}{ m\in r(\Z+\oh)}}
		r\log(1-z^{-1}e^{m})
\\=\quad&
\lim_{r\rightarrow 0}
	\sum_{\stackrel{C<m<D,}{ m\in r(a+b)(\Z+\oh)}}
		\sum_{k=0}^{a-1} r\log(1-z^{-1}e^{m+kr})
\\=\quad&
\frac{1}{a+b}
	\sum_{k=0}^{a-1}
		\lim_{r\rightarrow 0}
			\sum_{\stackrel{C<m<D,}{ m\in r(a+b)(\Z+\oh)}}
				r(a+b)\log(1-z^{-1}e^{m+kr})
\\\stackrel{\text{defn. of } \int}{=}&
\frac{1}{a+b}
	\sum_{k=0}^{a-1}
		\int_{C}^{D} \log (1-z^{-1}e^v)dv =
\frac{a}{a+b}
	\int_{C}^{D} \log (1-z^{-1}e^v)dv.
\end{align*}
The convergence is uniform in $z$ on compact sets of $\mathbb{C}\setminus [e^C,e^D]$.
\end{proof}

Using Lemma \ref{lem:CD} for each interval $[A_k,A_{k+1}], [A_{j-1},\tau],[\tau,A_j]$, and the fact that
\begin{equation*}
\lim_{rt \to \tau} r b(t) = B(\tau),
\end{equation*}
we get
\begin{align}
  S_{\tau,\chi}(z):=&\lim_{r\rightarrow 0} S^{(r)}_{t,h}(z)
= \nonumber 
-(\chi + B(\tau))\log z
\\&
+\sum_{k=1}^{j-1}
	\frac{a_k}{a_k+b_k}
		\int_{A_{k-1}}^{A_k}
			\log \bigl(1-z^{-1}e^v\bigr) dv
+\frac{a_j}{a_j+b_j}
	\int_{A_{j-1}}^\tau
		\log\bigl(1-z^{-1}e^v) dv
\\ \nonumber &
-\frac{b_j}{a_j+b_j}
	\int_{\tau}^{A_j}
		\log\big(1-ze^{-v}\bigr) dv
-\sum_{k=j+1}^{n}
	\frac{b_k}{a_k+b_k}
		\int_{A_{k-1}}^{A_k}
			\log \bigl(1-z e^{-v}\bigr)dv.
\end{align}

\section{Critical points of \texorpdfstring{$S_{\tau,\chi}(z)$}{S(z)}}

The localization of the critical points of the function $S_{\tau,\chi}$ is crucial in order to apply the steepest descent method. The following proposition gives a complete understanding of their nature.

\begin{Proposition} \label{prop:regions}
Fix $A_0 < \tau < A_n$. Then,

\begin{itemize}
  \item  there is a unique real number $\chi_0(\tau)$ such that
    \begin{enumerate}
      \item For any $\chi> \chi_0(\tau)$, $S_{\tau,\chi}(z)$ has exactly two non-real critical points, which are complex conjugates.
      \item For any $\chi<\chi_0(\tau)$, $S_{\tau,\chi}(z)$ has only real critical points.
    \end{enumerate}
  \item If $\tau\in (A_{j-1}, A_j)$ ,as $\chi\to +\infty$ the two non-real critical points $z_c,\bar{z}_c$ have the following asymptotic:
\begin{equation}\label{eq:crit_point_asymp}
  z_c,\bar{z}_c= e^{\tau} \exp\left( -e^{-\chi} \rho(\tau) e^{\pm i\pi\left(\frac{1}{2}+\alpha_j\right)}\right) (1+O(e^{-\chi})),
\end{equation}
where
\[
\rho(\tau)=\prod_{k=0}^n \left| 2 \sinh \left(\frac{\tau-A_k}{2}\right)\right|^{\alpha_{k+1}-\alpha_k},
\]
and $z_c$ is chosen to be the critical point with positive imaginary part.
\item If $\tau=A_j+\delta$, $\chi\to +\infty $ and $\delta\to 0$ such that
\[
p=e^{\chi-\chi^{(j)}}|\delta|^{1+\alpha_j-\alpha_{j+1}}
\]
is fixed, with
\begin{equation*}
e^{\chi^{(j)}} = \rho^{(j)} = \prod_{\substack{k=0 \\ k\neq j}}^{n} \Bigl| 2\sinh\bigl(\frac{A_j-A_k}{2}\bigr)\Bigr|^{\alpha_{k+1}-\alpha_k},
\end{equation*}
then the pair of complex conjugate solutions
behave as $z=e^{\tau-s|\delta|}$ where $s$ is a solution to the equation
\begin{equation}\label{eq:t-scaling}
  p=e^{\pm i\pi (\frac{1}{2}+\alpha_j)}\frac{(s-\mathrm{sign}(\delta))^{(\alpha_{j+1}-\alpha_j)}}{s}.
\end{equation}
\end{itemize}
\end{Proposition}

The proof is rather technical, and requires a number of intermediate statements.
It will be completed at the end of this section.

\subsection{An expression for \texorpdfstring{$\frac{d}{dz} S_{\tau,\chi}(z)$}{S'(z)}}

Using the identities
\begin{equation*}
 \frac{a_k}{a_k+b_k}=\frac{1}{2}+\alpha_k,\quad \frac{b_k}{a_k+b_k}=\frac{1}{2}-\alpha_k,
 \quad \log \bigl(\frac{1-z e^{-v}}{1-z e^{-u}}\bigr)=u-v + \log\bigl(\frac{z-e^{v}}{z-e^{u}}\bigr),
\end{equation*}
and the formula \eqref{eq:B} for $B(\tau)$ , one gets the following expression for the derivative of $S$ with respect to $z$:

\begin{align}
 z\frac{d}{dz}S(z) \nonumber =&
-(\chi+B(\tau))
+\sum_{k=1}^{j-1}
	\frac{a_k}{a_k+b_k}
	\int_{A_{k-1}}^{A_k}
		\frac{z^{-1}e^v}{1-z^{-1}e^v}dv
+\frac{a_j}{a_j+b_j}
\int_{A_{j-1}}^{\tau}
	\frac{z^{-1}e^v}{1-z^{-1}e^v}dv
\\ \nonumber &\quad
+\frac{b_j}{a_j+b_j}
\int_{\tau}^{A_j}
		\frac{z e^{-v}}{1-z e^{-v}}dv
+\sum_{k=j+1}^n
	\frac{b_k}{a_k+b_k}
	\int_{A_{k-1}}^{A_k}
		\frac{z e^{-v}}{1-z e^{-v}}dv
\\ \nonumber =&
-(\chi+B(\tau))
-\sum_{k=1}^{j-1}
	\frac{a_k}{a_k+b_k}
	\log \left( \frac{z-e^{A_k}}{z-e^{A_{k-1}}} \right)
-\frac{a_j}{a_j+b_j}
	\log \left( \frac{z-e^{\tau}}{z-e^{A_{j-1}}} \right)
\\ \nonumber &
+\frac{b_j}{a_j+b_j}
	\log \left( \frac{1-ze^{-A_j}}{1-ze^{-\tau}} \right)
+\sum_{k=j+1}^{n}
	\frac{b_k}{a_k+b_k}
		\log \left( \frac{1-ze^{-A_k}}{1-z e^{-A_{k-1}}} \right)
\\ \nonumber =&
-(\chi+B(\tau))
-\sum_{k=1}^{j-1}
	\frac{a_k}{a_k+b_k}
		\log\left( \frac{z-e^{A_k}}{z-e^{A_{k-1}}} \right)
-\frac{a_j}{a_j+b_j}
	\log\left(\frac{z-e^{\tau}}{z-e^{A_{j-1}}}\right)
\\ \nonumber &
+\frac{b_j}{a_j+b_j}
	\log\left( \frac{z-e^{A_j}}{z-e^{\tau}} \right)
+\sum_{k=j+1}^{n}
	\frac{b_k}{a_k+b_k}
		\log\left( \frac{z-e^{A_k}}{z- e^{A_{k-1}}} \right)
\\ \nonumber &
\underbrace{
	-\frac{b_j}{a_j+b_j}(A_{j}-\tau)
	-\sum_{k=j+1}^{n}
		\frac{b_k}{a_k+b_k}(A_{k}-A_{k-1})}
_{B(\tau)+\frac{\tau}{2}}
\\=&\label{eq:SPrime}
-(\chi-\frac{\tau}{2})
-\sum_{k=1}^{j-1}
	(\frac{1}{2}+\alpha_k)
	\log\left(\frac{z-e^{A_k}}{z-e^{A_{k-1}}}\right)
-(\frac{1}{2}+\alpha_j)
	\log\left(\frac{z-e^{\tau}}{z-e^{A_{j-1}}}\right)
\\ \nonumber &
+(\frac{1}{2}-\alpha_j)
	\log\left(\frac{z-e^{A_{j}}}{z-e^{\tau}}\right)
+\sum_{k=j+1}^{n}
	(\frac{1}{2}-\alpha_k)
		\log\left(\frac{z-e^{A_k}}{z-e^{A_{k-1}}}\right).
\end{align}

Note that each of the terms $\log\left(\frac{z-e^{A_k}}{z-e^{A_{k-1}}}\right)$ is a holomorphic function of $z$ on $\mathbb{C}\setminus[e^{A_{k-1}},e^{A_k}]$. Its imaginary part is in $(-\pi,\pi)$ and represents the angle between the vectors $z-e^{A_k}$ and $z-e^{A_k-1}$. This angle is 0 when $z\in\mathbb{R}\setminus[e^{A_{k-1}},e^{A_k}]$, and approaches $\pi$ (resp. $-\pi$) when $z$ approaches $(e^{A_{k-1}},e^{A_k})$ from above (resp. from below). See Figure~\ref{fig:angles}.

\subsection{Asymptotic of critical points when \texorpdfstring{$\chi$}{chi} goes to \texorpdfstring{$\infty$}{infinity}} \label{sec:crit_points_large}
Here we will prove the last two statements of Proposition \ref{prop:regions} concerning the asymptotics of the non real critical points when $\chi$ is large.

\begin{Lemma} \label{lem:largeXcrit}
Fix $A_0 < \tau < A_n$. Then, for sufficiently large $\chi$, $S_{\tau, \chi}(z)$ has exactly two critical points, which are complex conjugates and have non-zero imaginary parts. Their asymptotic as $\chi\to +\infty$ is described by the last
part of proposition \ref{prop:regions}.
\end{Lemma}
\begin{proof}
  \underline{First case:} $A_{j-1} < \tau < A_j$.

Critical points are solutions to the equation
\begin{equation}\label{eq:crit}
 z\frac{d}{dz}S(z) = 0.
\end{equation}

The real part of this equation can be written as
\begin{align*}
0=
\Re \left( z\frac{d}{dz}S(z) \right)
&=
-(\chi-\frac{\tau}{2})
-\sum_{k=1}^{j-1}
	(\frac{1}{2}+\alpha_k)
	\log\left|\frac{z-e^{A_k}}{z-e^{A_{k-1}}}\right|
-(\frac{1}{2}+\alpha_j)
	\log\left|\frac{z-e^{\tau}}{z-e^{A_{j-1}}}\right|
\\&
+(\frac{1}{2}-\alpha_j)
	\log\left|\frac{z-e^{A_{j}}}{z-e^{\tau}}\right|
+\sum_{k=j+1}^{n}
	(\frac{1}{2}-\alpha_k)
		\log\left|\frac{z-e^{A_k}}{z-e^{A_{k-1}}}\right|
\\&=
-(\chi-\frac{\tau}{2})
-\log\bigl|1-ze^{-\tau}\bigr|
+\sum_{k=0}^n
	(\alpha_{k+1}-\alpha_k)
		\log\bigl|1-z e^{-A_k}\bigr|.
\end{align*}

From here we can see that a solution $z$ as  $\chi\to +\infty$ should approach to $e^\tau$, or to $e^{A_l}$, for some $l\in\{1,\dots,n\}$ with $\alpha_{l+1}< \alpha_l$ (inner corner).

The imaginary part of equation (\ref{eq:crit}) for the critical points is:
\begin{multline}\label{eq:im-part}
0 =
-\sum_{k=0}^{j-1}
	(\frac{1}{2}+\alpha_k)
		\mathrm{angle}(z-e^{A_{k-1}},z-e^{A_k})
-(\frac{1}{2}+\alpha_j)\mathrm{angle}(z-e^{A_{j-1}},z-e^{\tau})
\\
+(\frac{1}{2}-\alpha_j)\mathrm{angle}(z-e^\tau,z-e^{A_j})
+\sum_{k=j+1}^{n}(\frac{1}{2}-\alpha_k)
	\mathrm{angle}(z-e^{A_{k-1}},z-e^{A_k}).
\end{multline}

\begin{figure}[h!]
  \begin{center}
  \includegraphics[width=7cm]{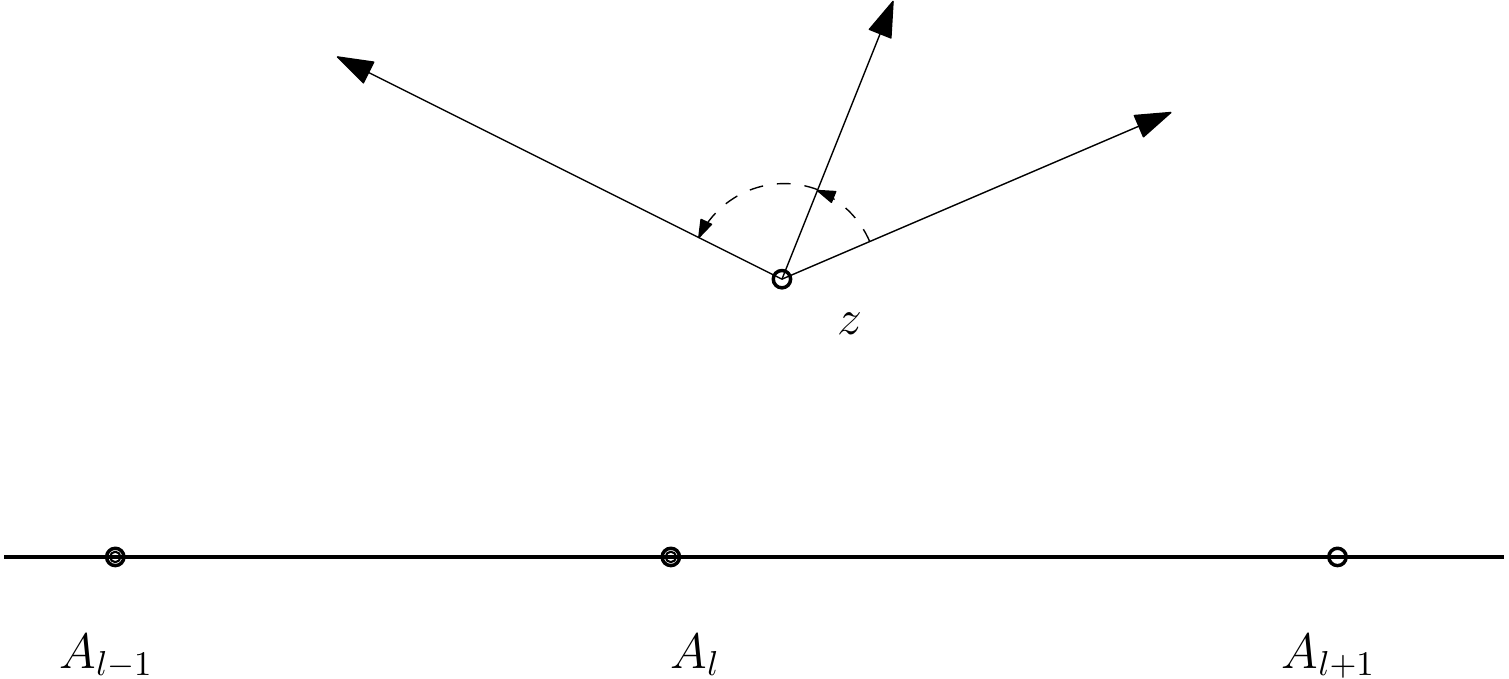}
  \caption{Geometric representation of the imaginary part of $\log\left( \frac{z-e^{A_l}}{z-e^{A_{l-1}}} \right)$ as the angle between $z-e^{A_l}$ and $z-e^{A_{l-1}}$.}
  \label{fig:angles}
  \end{center}
\end{figure}

Assume that $z$ is a critical point with positive imaginary part, and
$z\to e^{A_l}$ for some $l$ (since the equation has real coefficients,
the non real solutions will come as pairs of conjugate complex numbers).
Since $z$ approaches to the real axis, the angle between $z-e^{B}$ and
$z-e^C$ approaches to zero, unless $B$ or $C$ is equal to $A_l$.
Therefore as $\chi\to +\infty$ , two terms in Equation \eqref{eq:im-part}
will dominate. These are
\begin{equation*}
  \begin{cases}
    (\frac{1}{2}+\alpha_l)\mathrm{angle}(z-e^{A_{l-1}},z-e^{A_l}) + (\frac{1}{2}+\alpha_{l+1}) \mathrm{angle}(z-e^{A_l},z-e^{A_{l+1}}) & \text{if $\tau<A_l$},\\
    (\frac{1}{2}-\alpha_l)\mathrm{angle}(z-e^{A_{l-1}},z-e^{A_l}) +(\frac{1}{2}-\alpha_{l+1})
    \mathrm{angle}(z-e^{A_l},z-e^{A_{l+1}}) & \text{if $\tau > A_l$}.
  \end{cases}
\end{equation*}
Thus, these expressions should vanish as $\chi\to +\infty$.

But this is impossible, because the sum
\begin{equation*}
  \mathrm{angle}(z-e^{A_{l-1}},z-e^{A_l})+
  \mathrm{angle}(z-e^{A_l},z-e^{A_{l+1}})=
  \mathrm{angle}(z-e^{A_{l-1}},z-e^{A_{l+1}})
\end{equation*}
is positive and close to $\pi$ (see Figure \ref{fig:angles}). Therefore this option for a critical point is impossible.

The only possibility left is that $z\to e^{\tau}$ as $\chi\to +\infty$. Let $z=e^{\tau-\e}$, where $\e\to 0$ as $\chi\to \infty$ and $\arg\e\in(-\pi,\pi)$.

As $\e\to 0$, we have for $k\neq j$:
\begin{align*}
  \log\left(\frac{z-e^{A_k}}{z-e^{A_{k-1}}}\right) &=\log\left(\frac{e^{\tau}-e^{A_{k}}}{e^{\tau}-e^{A_{k-1}}}\right)+O(\e) \\
  &= \frac{A_k-A_{k-1}}{2}+\log 2\sinh\left| \frac{\tau-A_k}{2} \right|-\log 2\sinh\left| \frac{\tau-A_{k-1}}{2} \right|+O(\e), \\
  \log\left(\frac{z-e^{\tau}}{z-e^{A_{j-1}}} \right) &= \frac{\tau-A_{j-1}}{2}+\log(-\e)-\log  2\sinh\left|\frac{\tau-A_{j-1}}{2}\right| +O(\e),\\
  \log\left( \frac{z-e^{A_j}}{z-e^{\tau}} \right)&= \frac{A_j-\tau}{2}-\log(\e)+\log 2\sinh \left|\frac{\tau-A_j}{2}\right| +O(\e).
\end{align*}
Substitute these expressions in the equation for critical points:
\begin{align}
0 \nonumber &=
-(\chi-\frac{\tau}{2})
-\sum_{k=1}^{j-1}
	(\alpha_k+\frac{1}{2})
		\left(\frac{A_k-A_{k-1}}{2}
		+\log 2\sinh\Bigl| \frac{\tau-A_k}{2} \Bigr|
		-\log 2\sinh\Bigl| \frac{\tau-A_{k-1}}{2}\Bigr|\right)
\\ \nonumber &
-(\alpha_j+\frac{1}{2})
	\left(  \frac{\tau-A_{j-1}}{2}+\log(-\e)
	-\log 2\sinh\Bigl|\frac{\tau-A_{j-1}}{2}\Bigr|
	\right)
\\ \nonumber &
+(\frac{1}{2}-\alpha_j)
	\left(\frac{A_j-\tau}{2}+
	\log 2\sinh \Bigl|\frac{\tau-A_j}{2}\Bigr|
	-\log(\e)\right)
\\ \nonumber &
+\sum_{k=j+1}^{n}
	(\frac{1}{2}-\alpha_k)
		\left( \frac{A_k-A_{k-1}}{2}
		+\log 2\sinh\Bigl| \frac{\tau-A_k}{2} \Bigr|
		-\log 2\sinh\Bigl| \frac{\tau-A_{k-1}}{2} \Bigr|\right)
+O(\e)
\\&=
-\chi
-\log(\e)
\mp (\alpha_j+\frac{1}{2})i\pi
+ \sum_{k=0}^n
	(\alpha_{k+1}-\alpha_k)
	\log 2\sinh\Bigl|\frac{\tau-A_k}{2}\Bigr| 
+O(\e).
\label{eq:critpoints_multislope}
\end{align}
Here we used the identity 
\begin{equation*}
\log(-\varepsilon)=
  \begin{cases}
    \log \varepsilon -i\pi&\text{if $\arg\varepsilon >0$}, \\
    \log \varepsilon + i\pi&\text{if $\arg\varepsilon <0$},
  \end{cases}
\end{equation*}
which holds for our choice of branch for $\log$. The cancellation of the linear combination of the $A_i$ comes from the two possible expressions for $-B(\tau)$ given by Equations \eqref{eq:B} and \eqref{eq:B2}.

Solving this equation for $\e$, we arrive to the asymptotical formula:
\begin{align*}
\e = e^{-\chi} e^{\pm i\pi(\frac{1}{2}+\alpha_j)} \underbrace{\prod_{k=0}^n \left| 2 \sinh \left(\frac{\tau-A_k}{2}\right)\right|^{\alpha_{k+1}-\alpha_k}}_{\rho(\tau)}\left(1+O(e^{-\chi})\right),
\end{align*}
which implies that the asymptotic of the pair of complex conjugate solutions
is given by the formulae in the last part of proposition \ref{prop:regions}.

\underline{Second case:} $\tau$ becomes close to $A_j$ for $1\leq j\leq  N-1$ as $\chi$ goes to $+\infty$.

When $\tau=A_j+\delta$, $\delta\to 0$ and $\chi\to +\infty$,
the same arguments as above show that complex critical points still accumulate near $z=e^\tau$ but instead of the  Equation~\eqref{eq:critpoints_multislope} for $\e$, where $z=e^{\tau-\varepsilon}$ we will have

\begin{equation}
\label{eq:CritPtsNearCorn}
\chi-\chi^{(j)}=-(\alpha_j+\frac{1}{2})\log(-\varepsilon) 
+(\frac{1}{2}-\alpha_j)\log\left(\frac{\varepsilon-\delta}{\varepsilon}\right) - (\frac{1}{2}-\alpha_{j+1})\log(\varepsilon-\delta) +O(\e).
\end{equation}
From here we obtain equation (\ref{eq:t-scaling}) for $s=\varepsilon/|\delta|$.

This completes the proof of Lemma \ref{lem:largeXcrit}.
\end{proof}

Whereas the previous lemma describes what happens when $\chi$ goes to $+\infty$,  the situation when $\chi$ goes to $-\infty$ is given in the following:
\begin{Lemma} \label{lem:smallXcrit}
For each $A_0 < \tau < A_n$, and sufficiently negative $\chi$, every critical point of $S(z)$ is real.
\end{Lemma}

\begin{proof}
This is proven in the same way as Lemma \ref{lem:largeXcrit}. To have $\Re(z\frac{d}{dz}S(z))=0$, $z$ needs to be very close to $e^{A_l}$, for some $l\in\{0,1,\ldots,n\}$, such that $\alpha_{l+1}>\alpha_l$. Exactly as in the case $\chi\rightarrow\infty$, one can show that if $l\neq 0$ or $n$, then $\Im(z\frac{d}{dz}S(z))$ cannot be $0$.
If $z$ is near $e^{A_0}$ or $e^{A_n}$, then $z$ is actually real and respectively $z<e^{A_0}$ or $z>e^{A_n}$. This shows that when $\chi$ is sufficiently negative, $z\frac{d}{dz}S(z)=0$ has no complex solutions.
\end{proof}

\subsection{Proof of Proposition \ref{prop:regions}} Note that the last two parts of the proposition were proven in the previous section. To prove the first part we first show that for fixed $\tau\in (A_0, A_n)$ there is a unique value of $\chi=\chi_0(\tau)$ such that the number of critical points of $S(z)$ with non-zero imaginary part changes as $\chi$ passes through $\chi_0$. That is:

\begin{Proposition} \label{prop:one_double}
For each $\tau\in (A_0, A_n)$ there exists $\chi_0$ such that the number
of critical points of $S(z)$ with $\Im(z)\neq 0$ is the same for all
$\chi>\chi_0$, and it decreases by $2$ when $\chi$ passes the point
$\chi_0$ into the region  $\chi<\chi_0$.
\end{Proposition}

We prove Proposition \ref{prop:one_double} at the end of this section, after setting up some technical tools.

Define $f(z)=z\frac{d}{dz} S(z)$. It is clear that $z=0$ is not a critical point
of $S(z)$ and therefore critical points are zeroes of $f(z)$ in $\mathbb{C}\setminus\{0\}$. It is also clear
that $z$ is a double critical point if and only if
\[
f(z)=f'(z)=0.
\]

The following two lemmas give some information about the location of the double critical points of $S(z)$.

\begin{Lemma} \label{lem:noin}
There are no solutions to $f(z)=0$ in $[e^{A_0},e^{A_n}]$.
\end{Lemma}

\begin{proof}

Consider $z=x+i\epsilon$, where $x\in(e^{A_{i-1}},e^{A_i})$ and $|\epsilon|<<1$. Looking at the terms in $f(z)$ in formula (\ref{eq:SPrime}), it is easy to see that for such $z$ we have $$\Im\left(f(z)-c\log\left(\frac{z-e^{A_i}}{z-e^{A_{i-1}}}\right)\right) = O(\epsilon),$$ where $c\neq 0$ is the coefficient of $\log\left(\frac{z-e^{A_i}}{z-e^{A_{i-1}}}\right)$ in $f(z)$. Since $\log$ is the branch of the logarithm with an imaginary part in $(-\pi,\pi)$ with a cut along $\mathbb{R}_{-}$, we have that $\Im\left(c\log\left(\frac{z-e^{A_i}}{z-e^{A_{i-1}}}\right)\right) = \pm c \pi + O(\epsilon)$, which in turn implies $\Im(f(z)) = \pm c \pi + O(\epsilon)$. This shows that there are no solutions to $f(z)=0$ in $(e^{A_0},e^{A_n})$.

The only remaining points are $e^{A_j}$, but the function $f(z)$ is singular at
these points, and therefore they cannot be solutions either.
\end{proof}

\begin{Lemma} \label{lem:noout}
$f'(z)=0$ has exactly one real solution outside the interval $[e^{A_0},e^{A_n}]$ (possibly at $\infty$).
\label{prop:T(z)=0OneSoln}
\end{Lemma}

\begin{proof}
Recall that, by definition, the equation $f'(z)=0$ means
\begin{align}
f'(z) = \begin{cases}
\frac{d}{dz} f(z)=0 &  z \neq \infty \\
\lim_{z \rightarrow \infty} z^2 \frac{d}{dz} f(z)=0 & z=\infty.
\end{cases}
\end{align}
Thus solving $f'(z)=0$ in $\br  \cup \{ \infty \} \backslash [e^{A_0},e^{A_n}]$ is equivalent to solving $G(z)=0$ in this region, where
\begin{equation}
G(z) := (z-e^{A_j})(z-e^\tau)f'(z).
\end{equation}
When $\tau\in (A_{j-1}, A_j)$ we have
\begin{align}
&G(z)=\label{eq:G}
-\sum_{k=1}^{j-1}
	\left(\frac{1}{2}+\alpha_k\right) (e^{A_k}-e^{A_{k-1}})
		\frac{(z-e^{A_j})(z-e^\tau)}{(z-e^{A_k})(z-e^{A_{k-1}})} \\ \nonumber -&
\left(\frac{1}{2}+\alpha_j\right) (e^{\tau}-e^{A_{j-1}})
	\frac{(z-e^{A_j})(z-e^\tau)}{(z-e^{\tau})(z-e^{A_{j-1}})}
+\left(\frac{1}{2}-\alpha_j\right) (e^{A_j}-e^{\tau})
	\frac{(z-e^{A_j})(z-e^\tau)}{(z-e^{A_j})(z-e^{\tau})}
\\ \nonumber +&
	\sum_{k=j+1}^{n}
		\left(\frac{1}{2}-\alpha_k\right) (e^{A_k}-e^{A_{k-1}})
		\frac{(z-e^{A_j})(z-e^\tau)}{(z-e^{A_k})(z-e^{A_{k-1}})}.
\end{align}
Notice that each term has the form
\begin{equation} \nonumber
c\frac{(z-u_1)(z-u_2)}{(z-v_1)(z-v_2)}
\end{equation}
where $$c<0, u_1,u_2>v_1,v_2,\text{ OR }c>0, u_1,u_2<v_1,v_2,$$ and $$z>u_1,u_2,v_1,v_2,\text{ OR }z<u_1,u_2,v_1,v_2.$$

Under these conditions $$\frac{d}{dz}\left(c\log
\left(\frac{(z-u_1)(z-u_2)}{(z-v_1)(z-v_2)}\right)\right)= c\left(\frac{1}{z-u_1}+\frac{1}{z-u_2}-
\frac{1}{z-v_1}-\frac{1}{z-v_2}\right)<0,$$ for $z \in \br \backslash [e^{A_0},e^{A_n}]$,
so $G(z)$ is decreasing in $z$ for all $z \in \br \backslash [e^{A_0},e^{A_n}]$.
Furthermore, it is clear from formula \eqref{eq:G} that $G(z)$ is regular
at $\infty$. Thus $G(z)$ is decreasing as one moves along $(e^{A_n}, \infty]$,
and then ``wraps around infinity'' and moves from $-\infty$ to  $e^{A_0}$.

The function $G(z)$ can also be written as
\begin{equation*}
G(z)= (z-e^{A_j})(z-e^\tau) \left(
\frac{\frac{1}{2}+\alpha_1}{z-e^{A_0}}
-\sum_{k=1}^{n-1} \frac{\alpha_k-\alpha_{k+1}}{z-e^{A_k}}
+\frac{\frac{1}{2}-\alpha_n}{z-e^{A_n}}
-\frac{1}{z-e^\tau} \right).
\end{equation*}

Using this formula and the fact that $\frac{1}{2}+\alpha_1>0$ and $\frac{1}{2}-\alpha_n>0$, one can see that
\begin{equation*}
\lim_{z\to e^{A_0}-} G(z)<0, \lim_{z\to e^{A_n}+}G(z)>0.
\end{equation*}
this completes the proof of the Lemma.
\end{proof}

\begin{proof}[Proof of Proposition \ref{prop:one_double}]
As $\chi$ varies, the number of non-real critical points can change only
at the point $\chi_0$ when there is a double solution to $f(z)=0$ on
the real line, or at infinity. Lemma \ref{lem:noin} shows that this never
happens with $z \in [e^{A_0},e^{A_n}]$.

Lemma \ref{lem:noout} shows that $f'(z)=0$ has exactly one solution $z_0 \in \br \backslash [e^{A_0},e^{A_n}]$ (or possibly at $\infty$). Since $f(z)=g(z)+\chi$,
where $g(z)$ does not depend on $\chi$, it implies that there is exactly one
point $\chi_0=-g(z_0)$ where $f(z)=f'(z)=0$ and that this happens when $z=z_0$.
\end{proof}

\begin{proof}[Proof of Proposition \ref{prop:regions}]
By Lemma \ref{lem:largeXcrit}, for sufficiently large $\chi$ there are exactly
two critical points of $S_{\chi,\tau}(z)$ with $\Im(z)\neq 0$. By Lemma
\ref{lem:smallXcrit}, for sufficiently negative $\chi$ there are no such
critical points.  Taking into account Proposition \ref{prop:one_double},
this implies there is a unique value $\chi_0$ such that all critical points of
$S(z)$ are real for $\chi<\chi_0$ and there is a pair of complex conjugate
critical points $z_c, \bar{z_c}$ with $\Im(z_c)>0$ for $\chi>\chi_0$. This
completes the proof of Proposition \ref{prop:regions}. 
\end{proof}

\subsection{The asymptotic of critical points near corners}
The following proposition describes what happens in situations interpolating
the two regimes described in Proposition \ref{prop:regions}.

\begin{prop}\label{prop:near_corner}
  Complex conjugate solutions to (\ref{eq:t-scaling}) have the
following asymptotic when $p\to 0$:

\[
s=p^{-\frac{1}{1+\alpha_j-\alpha_{j+1}}}\exp\left(\frac{\pm i\pi(\frac{1}{2}+\alpha_j)}{1+\alpha_j-\alpha_{j+1}}\right)(1+O(p)).
\]

If $p\to \infty$ the asymptotic is different for different signs of $\delta$:
\begin{enumerate}
\item When $\delta\to +0$
\[
s=p^{-1}\exp\left(\pm i\pi(\frac{1}{2}+\alpha_{j+1})\right)\left(1+O\Bigl(\frac{1}{p}\Bigr)\right).
\]
\item When $\delta\to -0$
\[
s=p^{-1}\exp\left(\pm i\pi(\frac{1}{2}+\alpha_j)\right)\left(1+O\Bigl(\frac{1}{p}\Bigr)\right).
\]
\end{enumerate}
\end{prop}

\begin{proof}
We will do the case when $p\rightarrow \infty, \delta>0$. The other cases are similar. From (\ref{eq:t-scaling}) we see that if $p\rightarrow\infty$ then either $s\rightarrow 0$ or, if $(\alpha_{j+1}-\alpha_j)<0$, $s\rightarrow 1$. If $s\rightarrow 0$, then $s$ can be easily calculated to have the asymptotics given in the statement of the proposition. We need to show that the case $s\rightarrow 1$ is not possible. 

We have $z=e^{\tau-\varepsilon}, \varepsilon\rightarrow 0$, and we can assume that $Im(\varepsilon)\in(0,\pi)$.

Write $s$ as $s=1+te^{i\theta}$. Since $Im(\varepsilon)\in(0,\pi)$, we can assume $\theta\in(0,\pi)$ as well.

Let's rewrite (\ref{eq:CritPtsNearCorn}) as 
\begin{align*}
&\chi - \chi^{(j)}+(1+\alpha_j-\alpha_{j+1})\log(\delta)=
\\&=-(\alpha_j+\frac 12)\log(-s) + (\frac 12 - \alpha_j)\log(\frac{s-1}{s}) - (\frac 12 - \alpha_{j+1})\log(s-1) + O(\delta)
\\&=-(\alpha_j+\frac 12)\log(-1-te^{i\theta}) + (\frac 12 - \alpha_j)\log(\frac{te^{i\theta}}{1+te^{i\theta}}) - (\frac 12 - \alpha_{j+1})\log(te^{i\theta}) + O(\delta).
\end{align*}

Looking at the real part of this equation we get

\begin{equation*}
\chi - \chi^{(j)}+(1+\alpha_j-\alpha_{j+1})\log(\delta)
= (\frac 12 - \alpha_j)\log(t) - (\frac 12 - \alpha_{j+1})\log(t) + O(1).
\end{equation*}

From here we get

$$t=e^\frac{\chi-\chi^{(j)}+ (1+\alpha_j-\alpha_{j+1}) \log(\delta)}{\alpha_{j+1}-\alpha_j}\rightarrow 0.$$

Let's look at the imaginary part. We have
$$0 = -(\alpha_j+\frac 12)(-\pi+O(t))
 + (\frac 12 - \alpha_j)(\theta - O(t)) - (\frac 12 - \alpha_{j+1})\theta + O(\delta).$$

From here 

$$\theta=\pi\frac{\alpha_j+\frac 12}{\alpha_j-\alpha_{j+1}} + O(t) + O(\delta).$$

But this is impossible because $$\theta\in(0,\pi)$$ and $$0<\alpha_j-\alpha_{j+1}<\alpha_j+\frac 12.$$

\end{proof}

These asymptotics for $s$ as a function of $p$ imply the
following asymptotical formulae for complex conjugate critical
points $z=e^{\tau-\e}$of $S(z)$.

\begin{itemize}
\item When $\chi\to +\infty$ and $\delta\to 0$ such that
$e^\chi |\delta|^{1+\alpha_j-\alpha_{j+1}}\to 0$, we have
\[
\e=e^{\pm i\pi \frac{\alpha_j+\frac{1}{2}}{1+\alpha_j-\alpha_{j+1}}}
e^{-\frac{\chi}{1+\alpha_j-\alpha_{j+1}}}(1+o(1)).
\]

\item When $\chi\to +\infty$ and $\delta\to +0$, such that $p\to\infty $,
\[
\e=e^{\pm i\pi (\alpha_j+\frac{1}{2})}e^{-\chi}|\delta|^{\alpha_{j+1}-\alpha_j}
(1+o(1)).
\]

\item When $\chi\to +\infty$ and $\delta\to -0$, such that $p\to \infty $,
\[
\e=e^{\pm i\pi (\alpha_{j+1}+\frac{1}{2})}e^{-\chi}|\delta|^{\alpha_{j+1}-\alpha_j}
(1+o(1)).
\]
\end{itemize}

Notice that the limit $\chi\to \infty$ of complex conjugate critical points for fixed $\tau$ given by (\ref{eq:crit_point_asymp}) agrees with the last two asymptotics when $\tau\to A_j \pm 0$.

\section{Correlation functions, the frozen boundary, and the limit shape}\label{sec:corr}

In this section we will study correlation functions in the limit of the infinitely large system. In particular, the one-point correlation function gives the macroscopic density of horizontal rhombi.

In our analysis we will follow \cite{OR1} and \cite{OR2}. The basic idea is to use the steepest descent method for computing the asymptotical behavior of the double integral defining the correlation kernel.

It has been shown in \cite{OR1} that if $(\tau,\chi)$ is
such that all critical points of
$S(z)$ are real, the region in the vicinity of
this point is \emph{frozen}, i.e. in terms of tiling by rhombi
it is tiled with probability 1 by rhombi of one type
(tilted to the left, tilted to the right, or horizontal).

The region where two simple real critical points collapse into one
degenerate critical point is a curve in the $(\tau,\chi)$-plane
which separates the \emph{frozen region} from the \emph{disordered region}, where
the function $S(z)$ has a pair of complex conjugate simple
critical points.
This curve is the analogue of the arctic circle for tilings of large regular hexagons~\cite{CohnLarsenPropp}.

As it follows from Proposition \ref{prop:regions}, and in opposition to what occurs in \cite{OR1}, the boundary of the frozen region in our case projects bijectively to the interval $(A_0,A_n)$ on the $\tau$ axis. As we will see below, the region above this curve is disordered, the region below is
frozen.

\begin{figure}
  \begin{center}
  \includegraphics[width=10cm]{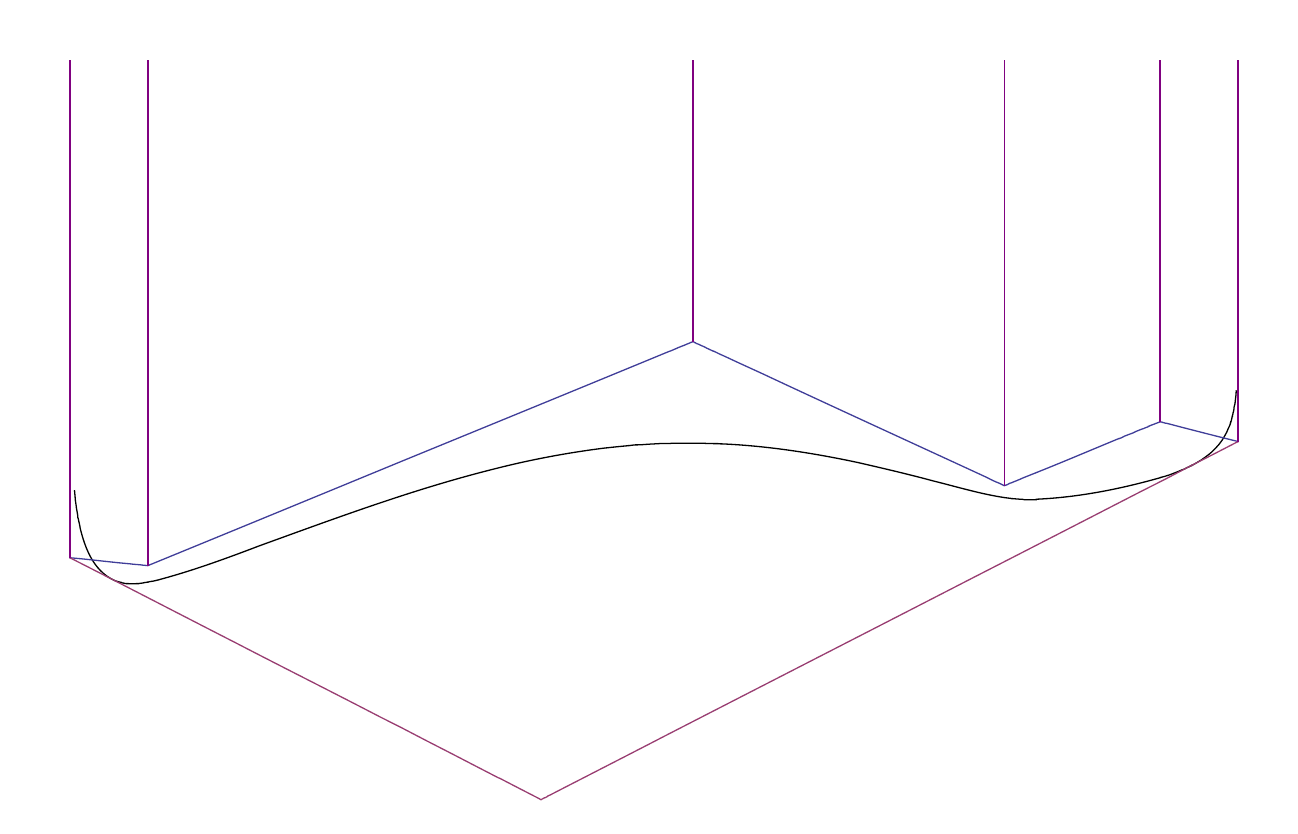}
   \caption{The frozen boundary when the parameters are \newline $A = \{-8, -7, 0, 4, 6, 7\}$, $\alpha = \{-0.1, 0.4, -0.45, 0.4, -0.25\}$ }
\label{fig:FrozenBoundary2}
  \end{center}
\end{figure}

\subsection{Local correlation functions in the bulk of the disordered region}\label{subsec:corr-disord}
In this section we will prove the following theorem.

\begin{Theorem}\label{thm:cor_beta}
In the limit when $r$ goes to $0$, the correlation functions of the system near a point $(\chi, \tau)$ in the bulk are given by determinants of the incomplete beta kernel
\begin{equation}\label{eq:lim_kernel}
  K_{\tau,\chi}^b(\Delta t, \Delta h) = \int_\gamma (1- e^{-\tau} z)^{\Delta t} z^{-\Delta h +\frac{\Delta t}{2}} \frac{dz}{2i\pi z},
\end{equation}
where the integration contour connects the two non-real critical points of $S_{\tau,\chi}(z)$, passing through the real line in the interval $(0,1)$ if $\Delta(t)\geq 0$ and through $(-\infty, 0)$ otherwise.
\end{Theorem}

The proof is completely parallel to the similar statements from \cite{OR1,OR2}. Here is the outline.

\begin{proof}
  The correlation functions are given by Formula (\ref{eq:main-corr2}). When $r\rightarrow 0$, the leading asymptotics of the integrand is given by $e^\frac{S_{\tau,\chi}(z)-S_{\tau,\chi}(w)}{r}\frac{1}{z-w}$. In the region where $S_{\tau,\chi}(z)$ has exactly two complex conjugate critical points we calculate the asymptotic of the integral by deforming the contours of integration as follows.

The poles of $\frac{\Phi_-(z,t_1)}{\Phi_+(z,t_1)}$ are real and lie in the interval $(e^{rt_1},\infty)$. The poles of $\frac{\Phi_+(w,t_2)}{\Phi_-(w,t_2)}$ are also real and lie in the interval $(0,e^{rt_2})$.
When $r\rightarrow 0$, they accumulate along the intervals  $(e^\tau,\infty)$ and $(0,e^\tau)$ respectively.

It is clear that any deformation of the contours of integration away from the real line doesn't change the integral, as long as the point where the $z$-contour crosses the positive real line does not move to the right, the point where the $w$-contour intersect the real line does not move to the left and the contours do not cross one another.

Note that whether the $z$-contour is inside the $w$-contour or vice-versa, depends on the sign of $\Delta t=t_2-t_1$. Now look at the level curves of $\Re(S(z))$ passing through $z_{c}$ and $\bar{z}_c$. In the case when the $z$-contour is outside, we can deform the original contours $C_z,C_w$, to $C'_z,C'_w$ as in Figure \ref{fig:ContourDeform}. The only poles we cross are the poles at $z=w$, so we need to take into account the contribution from the residues at $z=w$. We obtain \begin{multline} \label{eq:kernel_deformed_contours}
 K((t_1,h_1),(t_2,h_2))=
\int_{C'_z}\int_{C'_w}\text{same as in }(\ref{eq:main-corr2}) +\\
\int_{\mathcal{C}}\frac{\Phi_-(z,t_1)\Phi_+(z,t_2)}{\Phi_+(z,t_1)\Phi_-(z,t_2)} z^{h_2-h_1+B(t_2)-B(t_1)} \frac{dz}{2i\pi z},
\end{multline}
where $\mathcal{C}$ is an arc connecting points $\bar{z}_{c}$ and $z_{c}$ which crosses the real line on the positive side. In the case when the $w$-contour is inside of the $z$-contour, $\mathcal{C}$ crosses the real line on the  negative side.

\begin{figure}
\includegraphics[width=9cm]{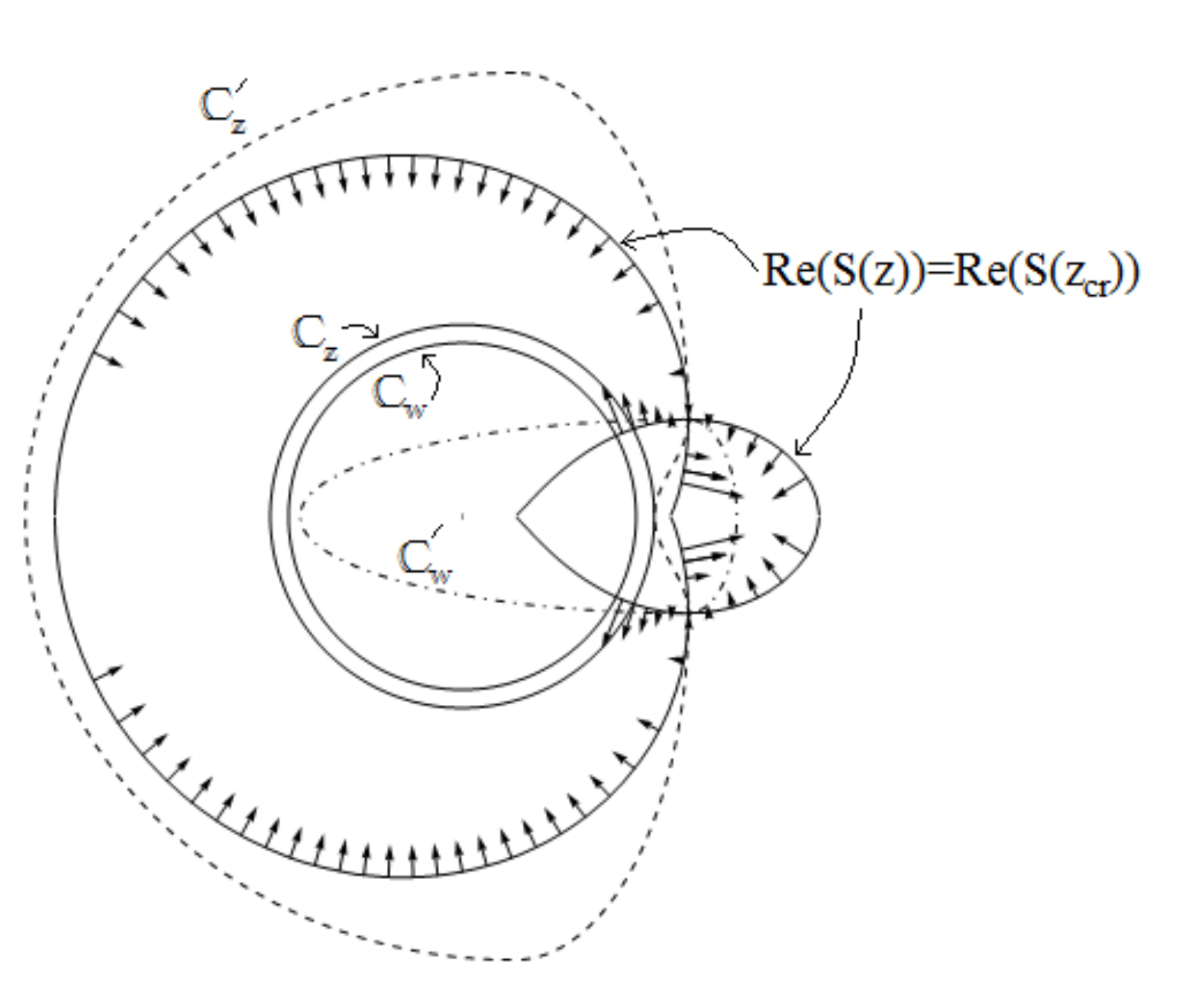}
\caption{  Deformation of contours when the $z$ contour is outside the $w$ contour.
}
\label{fig:ContourDeform}
\end{figure}

Now, let's show that in the limit $r\rightarrow 0$ the first integral is $0$. Since $\lim_{z\rightarrow 0} \Re S(z)=+\infty$, the gradient of the real part of $S(z)$ looks as in Picture \ref{fig:ContourDeform}, which implies that everywhere along the contours $C'_z,C'_w$ except when $z=w=$ critical points, we have $\Re(S(z))<\Re(S(w))$. Because the asymptotically leading term of the integrand is  $e^\frac{S_{\tau,\chi}(z)-S_{\tau,\chi}(w)}{r}\frac{1}{z-w}$, we conclude that the integral is $0$ as $r\to 0$ exponentially fast. Since $z-$ and $w-$ contours intersect transversally, there are no problems with the integrability of $\frac{1}{z-w}$ in the neighborhood of the critical points.

We conclude that the correlation kernel is equal to the second integral in \eqref{eq:kernel_deformed_contours}.

The second integral in \eqref{eq:kernel_deformed_contours}, i.e.
\begin{equation*}
  \int_{\mathcal{C}}\frac{\Phi_-(z,t_1)\Phi_+(z,t_2)}{\Phi_+(z,t_1)\Phi_-(z,t_2)} z^{h_2-h_1+B(t_2)-B(t_1)} \frac{dz}{2i\pi z},
\end{equation*}
is asymptotic, as $r\rightarrow 0$, $t_i r \rightarrow \tau$, $t_1-t_2=\Delta t$, $h_1-h_2=\Delta h$  to:

\begin{equation}\label{eq:lim_kernel0}
  \left(-e^\tau\right)^{N_-(t_1)-N_-(t_2)}\int_{\mathcal{C}}\left(1-e^{-\tau} z\right)^{\Delta t}z ^{-\Delta h + \frac{\Delta t }{2}}\frac{d z }{2i\pi z}
\end{equation}
where $N_-(t_i)$ is the number of descending unit pieces of back wall
between the leftmost corner $A_0$ and $t_i$. The prefactor in front of the
integral is of the form
\begin{equation*}
  \frac{g(t_1,h_1)}{g(t_2,h_2)},
\end{equation*}
which disappears by multi-linearity when we compute correlation functions,
as determinants of the form $\det(K( (t_i,h_i),(t_j,h_j) )   )$.

\end{proof}

The kernel $K_{\chi,\tau}^b$ has a simple expression on the diagonal: when
$\Delta t=\Delta h =0$, we have 
\begin{equation*} 
K_{\chi, \tau}^b(0,0)=\int_{\mathcal{C}}
\frac{dz}{2i\pi z}= \frac{\theta}{\pi},
\end{equation*} 
where $\theta$ is the argument  of
the critical point $z_c$ in the upper-half plane. This quantity is exactly the
density of horizontal tiles. It is therefore the vertical gradient of the limit
shape height function.

The correlation kernel depends only on one complex parameter
$z_c$. It is determined uniquely by coordinates $(\tau,\chi)$
or by the slope of the limit shape at this point.

\begin{cor}
  The distribution of the tiles in the neighborhood of a point $(\chi, \tau)$ converges to the translation invariant ergodic Gibbs measure on tilings of the plane with rhombi, with activities $1$, $e^{-\tau}|z_c|$ and $|z_c|$ for the three types of tiles in absence of a magnetic field.
\end{cor}

Indeed, from \cite{Ke:LocStat,KOS} the correlation functions for the translation invariant probability measure with activities $a$, $b$, $c$ on rhombus configurations (with no external field) are given by determinants of the correlation kernel
\begin{equation*}
  K_{a:b:c}((t_1,h_1), (t_2,h_2)) = \iint_{{\mathbb T}^2} \frac{z^{-\Delta h+ \frac{\Delta t}{2}}w^{\Delta t}}{a+\frac{1}{w}(b+\frac{c}{z})} \frac{dz}{2i\pi z} \frac{dw}{2i\pi w}.
\end{equation*}
Taking the proposed choice for the activities, and computing the integral over $w$ by residues leads to a simple integral along a sector of the unit circle. After a change of variable in the integral by multiplying by $|z_c|$, one gets the same expression as in \eqref{eq:lim_kernel}, see \cite{OR1}. The convergence of finite dimensional distributions follows from Theorem \ref{thm:cor_beta}. Since the space of tilings is compact for the product topology, the convergence of finite dimensional distributions implies the convergence of the whole distribution.

This result gives a precise description of the local behavior of the system in the neighborhood of a point inside the liquid region of the limit shape. However, the expressions of the critical points $z_c$ and $\overline{z_c}$ are not explicit for a finite $\chi$. In the next subsection, we investigate the asymptotic when $\chi \to \infty$ .

Similar arguments can be applied to the case when all critical points are real. In this case the integral (\ref{eq:main-corr2})
tends either to zero or to one depending on the sign of the critical points. The steepest descent method results in the
following asymptotic:
\[
\rho_{(t_1,h_1),\dots, (t_k,h_k)}=\rho+ O(e^{-\frac{\alpha}{r}}),
\]
where $\alpha>0$,  $\rho=1$ when $(\tau,\chi)$ is below the segment of the boundary of the limit shape between two points where it touches the boundary, and $\rho=0$ when $(\tau,\chi)$ is below the boundary of the limit shape but above the turning points.

\subsection{The correlation functions for large \texorpdfstring{$\chi$}{chi} and the bead model}
In this section, we investigate the behavior of the high piles of cubes of the random skew plane partition, that are close to the back wall.
The asymptotic of local correlation functions as $\chi\to \infty$ depends on whether $\tau$ is in a vicinity of $A_j$ or not.

\subsection{\texorpdfstring{$A_{j-1}<\tau <A_j$}{A\_\{j-1\}<tau<A\_\{j\}}}
In this subsection we consider the case when $A_{j-1}<\tau <A_j$.

  Let 
\begin{equation*}
 t_i = \frac{\tau}{r}+\eta_i, \quad h_i = \frac{\chi}{r}+e^{\chi}\xi_i, \quad i=1,2,
\end{equation*}
so that $\Delta \eta = \Delta t$ and $\Delta \xi = e^{-\chi} \Delta h$.
\begin{thm}
  In the limit where $\chi\to \infty$, such that $\Delta \eta$ and $\Delta\xi$ are fixed, the kernel $K_{\tau,\chi}^b( (t_1,h_1),(t_2,h_2))$ of Equation \eqref{eq:lim_kernel} divided by $e^{\chi}$ and a factor which does not affect the determinant (\ref{eq:det}), converges to
\begin{equation*}
K^{(\gamma)}((\eta_1,\xi_1),(\eta_2,\xi_2))=
\begin{cases}
  {\displaystyle d\int_{-1}^{1} \left(\gamma +i \varphi \sqrt{1-\gamma^2}\right)^{\eta_1-\eta_2} e^{-i\varphi d (\xi_1-\xi_2)} \frac{d \varphi}{2\pi}}, & \eta_1 \geq \eta_2,\\
   \displaystyle -d \int_{\mathbb{R}\setminus[-1,1]} \left(\gamma +i \varphi \sqrt{1-\gamma^2}\right)^{\eta_1-\eta_2} e^{-i\varphi d (\xi_1-\xi_2)} \frac{d \varphi}{2\pi}, & \eta_1 < \eta_2,
 \end{cases}
\end{equation*}
where $\alpha_j$ is the slope of the piece of the back wall  between $A_{j-1}$ and $A_j$, $\gamma=\sin \pi \alpha_j$, and 
\begin{equation*}
 d=\rho(\tau)\cos(\pi\alpha_j)=\cos(\pi\alpha_j)\prod_{k=0}^n \left| 2 \sinh \left(\frac{\tau-A_k}{2}\right)\right|^{\alpha_{k+1}-\alpha_k}.
\end{equation*}
\end{thm}

\begin{proof}
Suppose  $t_1 \geq t_2$.
Recall that the two complex conjugate critical points $z_c$ and $\bar{z}_c$ for this value of $\tau$ have the asymptotics \eqref{eq:crit_point_asymp} when $\chi\to \infty$. The curve $\mathcal{C}$ joining $z_c$ and $\bar{z}_c$ can be chosen as the positively oriented arc of the circle centered at 0, of radius $|z_c|$. A possible parametrization of this arc is:
\begin{equation*}
 z= z(\varphi)= e^ \tau \exp\left(e^{-\chi} \rho(\tau) \left(\sin(\pi\alpha_j)+ i\varphi \cos(\pi\alpha_j)\right)\right)(1+o(1)),
\end{equation*}
where $\varphi$ runs from $-1$ to $+1$.

The expression of the kernel \eqref{eq:lim_kernel} becomes
\begin{align*}
  K_{\chi,\tau}^b(\Delta t, \Delta h) &= e^{-\chi}\rho(\tau) \cos(\pi \alpha_j)
\int_{-1}^{1} \left ( 1-e^{-\tau} z(\varphi)\right )^{\Delta t} \bigl( z(\varphi)\bigr)^{-\Delta h + \Delta t/2}\frac{d\varphi}{2\pi}.
\end{align*}

When $\chi$ goes to infinity, one has
\begin{align*}
 \left( 1-e^{-\tau} z(\varphi)\right)^{t_1-t_2} &= (-e^{-\chi}\rho(\tau))^{\eta_1-\eta_2} \left(\sin(\pi \alpha_j) + i\varphi \cos(\pi \alpha_j)\right)^{\eta_1-\eta_2} (1+o(1)),\\
\bigl( z(\varphi)\bigr)^{-(h_1-h_2)} &= \left(e^{\tau e^\chi}e^{\rho(\tau)\sin(\pi\alpha_j)}\right)^{-(\xi_1-\xi_2)} e^{-i\rho(\tau)\varphi \cos(\pi\alpha_j)(\xi_1-\xi_2)} (1+o(1)).
\end{align*}

The factors
\begin{equation*}
 (e^{-\chi}\rho(\tau))^{\eta_1-\eta_2}\left(e^{\tau e^\chi}e^{\rho(\tau)\sin(\pi\alpha_j)}\right)^{-(\xi_1-\xi_2)}
\end{equation*}
are of the form $\frac{g(\xi_1,\eta_1)}{g(\xi_2,\eta_2)}$. These factors have no effect on the computations of the probabilities, since they cancel out when computing, in the limit, the determinant $\det(K_{\chi,\tau}( t_i-t_j,h_i-h_j))$. We can thus drop them out, and obtain a new kernel defining the same determinantal process.

The integral converges as $\chi \rightarrow +\infty$ to
\begin{equation*}
 \rho(\tau)\cos(\pi \alpha_j)\int_{-1}^{1} \left(\sin(\pi \alpha_j) +i \varphi \cos(\pi\alpha_j)\right)^{\eta_1-\eta_2} e^{-i\varphi \rho(\tau) \cos(\pi\alpha_j)(\xi_1-\xi_2)} \frac{d \varphi}{2\pi}.
\end{equation*}

Similarly, when $t_1 < t_2$, we use the same change of variable $z=z(\varphi)$. But now, $\varphi\in ( -\frac{\pi e^\chi}{\rho(\tau)\cos(\pi\alpha_j)} , -1)\cup ( 1, \frac{\pi e^\chi}{\rho(\tau)\cos(\pi \alpha_j)})$. The same computations as above show that $K_{\chi,\tau}(\Delta t, \Delta h)$ divided by $\frac{g(\xi_1,\eta_1)}{g(\xi_2,\eta_2)}$ converges to
\begin{equation*}
 -\rho(\tau)\cos(\pi \alpha_j)\int_{\mathbb{R}\setminus[-1,1]} \left(\sin(\pi \alpha_j) +i \varphi \cos (\pi \alpha_j)\right)^{\eta_1-\eta_2} e^{-i\varphi \rho(\tau) \cos(\pi\alpha_j )(\xi_1-\xi_2)} \frac{d \varphi}{2\pi}.
\end{equation*}

The kernel divided by $e^{-\chi}$ and some factors $\frac{g(\eta_1,\xi_1)}{g(\eta_2,\xi_2)}$ which cancel in the determinant, converges to
\begin{equation*}
 \begin{cases}
  \displaystyle \rho(\tau)  \cos(\pi\alpha_j) \int_{-1}^{1} \left(\sin(\pi\alpha_j)+i\varphi\cos(\pi\alpha_j)\right)^{\eta_1-\eta_2} e^{-i\cos(\pi\alpha_j)\rho(\tau)(\xi_1-\xi_2)\varphi} \frac{d\varphi}{2\pi}, &  t_1\geq t_2 \\
  \displaystyle -\rho(\tau)  \cos(\pi\alpha_j) \int_{\mathbb{R}\setminus[-1,1]}\!\!\! \left(\sin(\pi\alpha_j)+i\varphi\cos(\pi\alpha_j)\right)^{\eta_1-\eta_2} e^{-i\cos(\pi\alpha_j)\rho(\tau)(\xi_1-\xi_2)\varphi} \frac{d\varphi}{2\pi},  & t_1 < t_2
 \end{cases}
\end{equation*}

This ends the proof. \end{proof}

As a consequence, using the same argument as in \cite{beads}, we have the following theorem
\begin{thm}
 The point process describing the positions of horizontal rhombi in the neighborhood of a non frozen point $(\tau,\chi)$, with $A_{j-1}<\tau <A_j$ in the limit shape $r\rightarrow 0$ converges to the bead process on $\mathbb{Z}\times \mathbb{R}$ with parameter $\gamma=\sin\pi\alpha_j$ and density
\begin{equation*}
  \frac{1}{\pi}\rho(\tau)\cos(\pi\alpha_j) = \frac{1}{\pi} \cos(\pi\alpha_j)\prod_{k=0}^n\left| 2 \sinh\left(\frac{\tau-A_{k}}{2} \right)\right|^{\alpha_{k+1}-\alpha_k}.
\end{equation*}
\end{thm}

Note that the density is singular when at the corners: it tends to $0$ at inner corners (where $\alpha_{k+1} > \alpha_k$), and goes to $\infty$ at outer corners. This is a remnant of the singularities of the limit shape observed in \cite{OR1}.
This suggests that the scaling has to be modified to observe a non trivial phenomenon in the vicinity of the corners $A_k$. 

\subsubsection{$\tau$ is in the vicinity of $A_j$ for $0<j<n$}
Now assume that $\tau=A_j+\delta$, and $\delta\to 0$ as $\chi\to \infty$ in such a way that $p=e^{\chi-\chi^{(j)}}|\delta|^{1+\alpha_j-\alpha_{j+1}}$ is fixed. Recall that in this case, the critical point of $S_{\tau,\chi}(z)$ are given by Equation~\eqref{eq:t-scaling}. 
In this limit the curve $\mathcal{C}$ joining two complex conjugate critical points can be parameterized as
\[
z(\phi)=e^\tau \exp(-t^{\beta_j}e^{-\beta_j(\chi-\chi^{(j)})}(s'+is''\phi)),
\]
where $s'$ and $s''$ are real and imaginary parts of $s$ respectively, and $\beta_j=\frac{1}{1+\alpha_j-\alpha_{j+1}}$.

\begin{thm} \label{thm:scaling}Assume that $\chi\to \infty$ and $\delta\to 0$ as above, and that $\Delta\eta=\Delta t$ and $\Delta\zeta=e^{-\beta_j\chi}\Delta h$ are fixed. Then the correlation kernel \eqref{eq:lim_kernel}, divided by $e^{-\beta_j\chi}$ and some factors not affecting the determinant,  converges to the same expression as in the previous theorem 
\begin{equation*}
K^{(\gamma')}((\eta_1,\zeta_1),(\eta_2,\zeta_2))=
\begin{cases}
 d' \int_{-1}^{1} \left(\gamma' +i \varphi \sqrt{1-\gamma'^2}\right)^{\eta_1-\eta_2} e^{-i\varphi d' (\zeta_1-\zeta_2)} \frac{-d \varphi}{2\pi} & \text{if $\eta_1 \geq \eta_2$}\\
  - d' \int_{\mathbb{R}\setminus[-1,1]} \left(\gamma' +i \varphi \sqrt{1-\gamma'^2}\right)^{\eta_1-\eta_2} e^{-i\varphi d' (\zeta_1-\zeta_2)} \frac{d \varphi}{2\pi} & \text{if $\eta_1 < \eta_2$},
 \end{cases}
\end{equation*}
with
\[
 d'=e^{\beta_j\chi^{(j)}}s'', \quad \gamma'=-\frac{s''}{|s|}.
\]
\end{thm}
The proof is completely parallel to the proof of the previous statement, so we will skip it.

It is easy to see that these formulae when $\delta\to\pm 0$ and $p\to \infty$ agree with the previous result when $\tau\to A_j \pm 0$. let us verify it for $\delta\to -0$, the calculations being almost identical for $\delta\to +0$. We have in the limit $\delta\rightarrow 0-$:
\[
s\sim p^{-1}e^{i\pi (\alpha_j+\frac{1}{2})}.
\]
From the definition of $p$ and a simple algebra we have 
\begin{equation*}
  ( p e^{-\chi})^{\beta_j-1}=|\delta|^\frac{\beta_j-1}{\beta_j}
e^{-(\beta_j-1)\chi^{(j)}}.
\end{equation*}

Taking into account the definition of $\Delta\xi$ and $\Delta\zeta$ we obtain:
\[
\Delta\zeta e^{\beta_j\chi^{(j)}}s'' \sim |\delta|^{\alpha_{j+1}-\alpha_j}e^{\beta_j\chi^{(j)}}\cos(\pi\alpha_j)
\Delta\xi.
\]
Moreover, when $\tau\to A_j-0$
\begin{equation*}
  \rho(\tau) \sim \rho^{(j)} |\delta|^{\alpha_{j+1}-\alpha_j} = |\delta|^{\alpha_{j+1}-\alpha_j}e^{\chi^{(j)}},
\end{equation*}
where $\delta=\tau-A_j$. Comparing the two kernels in these regimes yields:

\[
e^{-\beta_j\chi} K^{(\gamma')}((\eta_1,\zeta_1),(\eta_2,\zeta_2))\sim
e^{-\chi} K^{(\gamma)}((\eta_1,\xi_1),(\eta_2,\xi_2)).
\]

It is also easy to find the asymptotic when $t\to 0$.
In this case:
\[
\gamma'\to \sin\left(\frac{\pi}{2}\frac{\alpha_{j+1}+\alpha_j}{1+\alpha_j-\alpha_{j+1}}\right),
\]
and for the density we have:
\[
e^{\beta_j\chi^{(j)}}t^{\beta_j}s''\to e^{\beta_j\chi_j}\cos(\frac{\pi}{2}\frac{\alpha_{j+1}+\alpha_j}{1+\alpha_j-\alpha_{j+1}}).
\]

Again, as a corollary of Theorem \ref{thm:scaling} 
the statistics of the horizontal tiles, rescaled vertically by a factor $e^{-\chi\beta_j}$, converge as $\chi\to\infty$  to the bead process with parameter $\gamma=\sin\left(\frac{\pi(\alpha_{j+1}+\alpha_j)}{2(1+\alpha_j-\alpha_{j+1})}\right)$ 
and density
$\frac{1}{\pi} \left(\rho^{(j)}\right)^{\beta_j}\cos\left(\frac{\pi(\alpha_{j+1}+\alpha_j)}{2(1+\alpha_j-\alpha_{j+1})}\right)$.

Notice that in the vicinity of $\tau=A_j$  one can get a bead process with any parameter in the interval $[\sin(\pi\alpha_j), \sin(\pi \alpha_{j+1})]$ by appropriately tuning $p$.


\begin{thebibliography}{E-G-S}

\bibitem[B]{beads} Cedric Boutillier, The bead model and limit behaviors of dimer models. \textit{Ann. of Probab.} \textbf{37} (2009), Number 1, Page 107--142.

\bibitem[CK]{CerfKenyon} Rapha\"el Cerf and Richard Kenyon. The
    low-temperature expansion of the Wulff crystal in the 3D Ising model.
    \textit{Comm. Math. Phys.} \textbf{222} (2001), no. 1, 147--179. 

\bibitem[CLP]{CohnLarsenPropp} H. Cohn and M. Larsen and J. Propp, The shape of a typical boxed plane partition. \textit{New York J. Math.} \textbf{4} (1998), Page 137--165.

\bibitem[G]{G} V. Gorin, Nonintersectiong paths and the Hahn orthogonal polynomial ensemble. \textit{Functional Analysis and Its Applications.} \textbf{42} (2008), Number 3, Page 180--197.

\bibitem[J]{J} K. Johansson, Non-intersecting paths, random tilings and random matrices. \textit{Probab. Theory and Related Fields} \textbf{123} (2002), Number 2, Page 225--280.

\bibitem[Ka]{Ka} P. W. Kasteleyn,  The statistics of dimers on a lattice. \textit{Physica} \textbf{27} (1961), Page 1209--1225.

\bibitem[Ke]{Ke:LocStat} Kenyon, Richard. Local statistics of lattice dimers. {\it Ann. Inst. H. Poincar{\'e} Probab. Statist.} \textbf{33} (1997), Number 5, Pages 591--618.

\bibitem[Ke2]{Ke:GFF2} Kenyon, Richard. Height fluctuations in the honeycomb dimer model. \textit{Comm. Math. Phys.} (2008).


\bibitem[KO]{KO} Kenyon, Richard and Okounkov Andrei, Limit shapes and the complex Burgers equation, \textit{Acta Math.} \textbf{199} (2007), Number 2, Page 263--302. 

\bibitem[KOS]{KOS} Kenyon, Richard; Okounkov, Andrei; Sheffield, Scott Dimers and amoebae. Ann. of Math. (2) 163 (2006), no. 3, 1019--1056.
 
\bibitem[MCW]{MCW} B. McCoy, and F. Wu, \textit{The two-dimensional Ising model}, Harvard Univ. Press, 1973.

\bibitem[M]{M} Sevak Mkrtchyan, Scaling limits of random skew plane partitions with arbitrarily sloped back walls.  {\it Comm. Math. Phys.} (2011), DOI 10.1007/s00220-011-1277-y.


\bibitem[NHB]{NHB}  Nienhuis, B.; Hilhorst, H. J.; Bl\"ote, H. W. J. Triangular SOS models and cubic-crystal shapes. J. Phys. A 17 (1984), no. 18, 3559--3581.

\bibitem[OR1]{OR1} Okounkov, Andrei and Reshetikhin, Nikolai. Correlation function of Schur process with application to local geometry of a random 3-dimensional Young diagram. {\it J. Amer. Math. Soc.}
\textbf{16} (2003), 581--603 (electronic).

\bibitem[OR2]{OR2} Okounkov, Andrei and Reshetikhin, Nikolai. Random skew plane partitions and the Pearcey process. {\it Commun. Math. Phys.}, \textbf{269}, (2007), 571–609.

\bibitem[OR3]{OR3} Okounkov, Andrei and Reshetikhin, Nikolai. The birth of a random matrix. {\it Moscow Mathematical Journal} \textbf{6}, Number 3. July-September 2006, Pages 553--566.

\bibitem[TF]{TempFish} H. N. V. Temperley and M. E. Fisher. Dimer problem in statistical mechanics -- an exact result.  \textit{Philos. Mag.} \textbf{6} (1961), Pages 1061--1063.

\end{thebibliography}
\end{document}